\newcommand{\D}[1]{\text{d}#1}
\newcommand{\sign}{\text{sign\,}}
\documentclass[%
 reprint,
superscriptaddress,
 amsmath,amssymb,
 aps,
]{revtex4-2}

\def\Xint#1{\mathchoice
   {\XXint\displaystyle\textstyle{#1}}%
   {\XXint\textstyle\scriptstyle{#1}}%
   {\XXint\scriptstyle\scriptscriptstyle{#1}}%
   {\XXint\scriptscriptstyle\scriptscriptstyle{#1}}%
   \!\int}

\def\XXint#1#2#3{{\setbox0=\hbox{$#1{#2#3}{\int}$}
     \vcenter{\hbox{$#2#3$}}\kern-.5\wd0}}
\def\dashint{\Xint-}
\usepackage{amsthm} 
\usepackage{hyperref}
\hypersetup{
    colorlinks=true,
    linkcolor=blue,
    filecolor=magenta,      
    urlcolor=cyan,
}
\usepackage{tikz}
\usepackage{pgfplots}
\usepackage{color}
\newtheorem{property}{Lemma}
\def\doi{http://dx.doi.org/}
\usepackage{braket}

\usepackage{graphicx}
\usepackage{dcolumn}
\usepackage{bm}

\begin{document}


\title{Quantum computing critical exponents}


\author{Henrik Dreyer}
\affiliation{Rudolf Peierls Centre for Theoretical Physics, Clarendon Laboratory, Oxford OX1 3PU, UK}
\affiliation{Cambridge Quantum Computing Ltd.
9a Bridge Street, CB2 1UB, Cambridge,
United Kingdom}
\author{Mircea Bejan}
\affiliation{Rudolf Peierls Centre for Theoretical Physics, Clarendon Laboratory, Oxford OX1 3PU, UK}
\author{Etienne Granet}
\affiliation{Rudolf Peierls Centre for Theoretical Physics, Clarendon Laboratory, Oxford OX1 3PU, UK}

\date{\today}

\begin{abstract}
We show that the Variational Quantum-Classical Simulation algorithm admits a finite circuit depth scaling collapse when targeting the critical point of the transverse field Ising chain.  The order parameter only collapses on one side of the transition due to a slowdown of the quantum algorithm when crossing the phase transition. In order to assess performance of the quantum algorithm and compute correlations in a system of up to 752 qubits, we use techniques from integrability to derive closed-form analytical expressions for expectation values with respect to the output of the quantum circuit. We also reduce a conjecture made by Ho and Hsieh \cite{hoetal} about the exact preparation of the transverse field Ising ground state to a system of equations.
\end{abstract}

\maketitle


\section{\label{sec:level1}Introduction}
Impressive advances have recently led to the realisation of the first noisy intermediate scale quantum (NISQ) computers \cite{preskill,arute, zhong} and quantum simulators \cite{blatt, zhang, islam, islam2, greiner, bloch, Bernien}. One promising application of NISQ devices is the simulation of quantum many-body systems \cite{feynman, cirac,Aidelsburger,choi, Aidelsburger2, choi2, islam3}. Near-term quantum computers will be limited by the number of qubits and the number of gates that can be executed with high fidelity, while analog simulations have to be executed within a typically short coherence time. These restrictions make it challenging to map out phase diagrams of strongly correlated materials, in particular in the vicinity of quantum critical points. At these phase transitions, the correlation length diverges and representing the system of interest with a finite device necessarily comes with a loss of accuracy. Instead of pushing computational resources towards the thermodynamic limit, classical methods typically make use of data produced by smaller scale computations and combine those data points in an informed manner, as it is done in finite size
\cite{fisherbarber,brezin,cardybook,ising_collapse}
or finite bond dimension scaling collapses \cite{scaling_collapse_01}.

Similarly, one may ask if data produced by a NISQ quantum computer can be used to predict the location and universality class of a critical point with an accuracy that goes beyond the machine specifications.

In this work, we provide an example where this is indeed possible. More precisely, we show that the Variational Quantum Classical Simulation (VQCS) \cite{hoetal} algorithm admits a scaling collapse when targeting the critical point of the transverse field Ising model (TFIM). Instead of finite size or finite entanglement, the depth of the quantum circuit plays the role of the relevant perturbation away from criticality. To classically benchmark the performance of the quantum algorithm, we adapt techniques from integrability to the quantum computation setting. Specifically, the quantum circuit can be mapped to a sequence of quenches in the TFIM. We show how to describe the state after a series of such quenches and derive a closed-form analytical expression for expectation values of local Gaussian observables (like the energy) in a circuit of arbitrary depth and width. This allows us to optimize the parameters of the quantum circuit for a number of qubits that is inaccessible to other classical simulation techniques but within reach of a NISQ computer.

Importantly, assessing the physics of the optimal state requires the evaluation of expectation values that are non-Gaussian or non-local. Our key technical contribution is a framework in which \emph{arbitrary} observables can be evaluated classically in polynomial time, provided the circuit is Gaussian.  Using this new framework, we compute the order parameter of the optimal circuit. We find that finite-circuit-depth data of the order parameter collapses on one side of the phase transition, in the phase that contains the initial state of the quantum computer.
The framework we develop can also be used as a point of departure for quantum circuit design.

\begin{figure}[ht]
\begin{center}
 \includegraphics[scale=0.13]{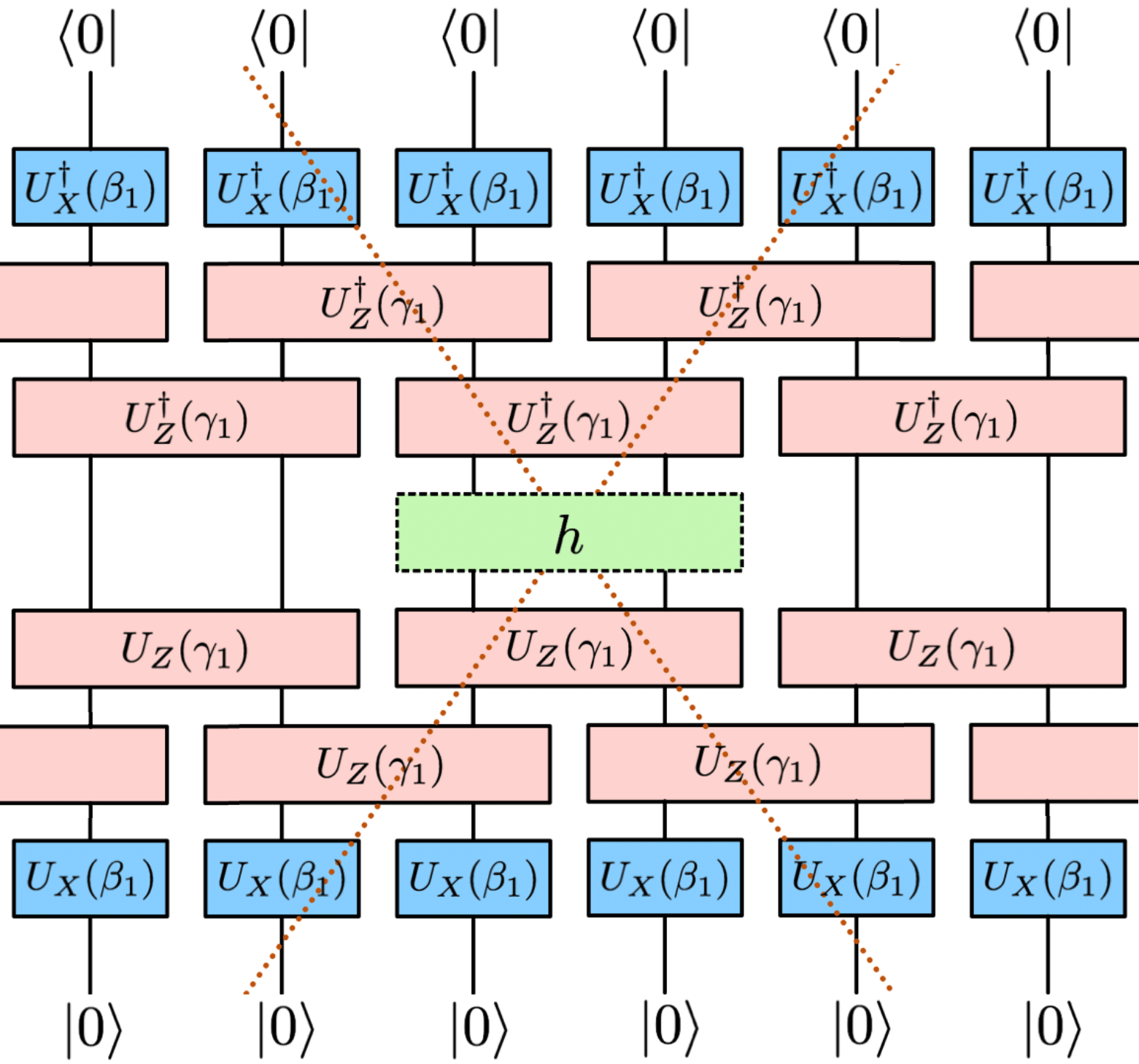}
 \caption{The quantum circuit~(\ref{eq_algo}) for $p=1$. The expectation value of the energy density (green) involves only those qubits that are within the lightcone (dotted). Here, $U_X(\beta) = \exp(i \beta X)$, $U_{ZZ}(\gamma) = \exp(i \gamma Z \otimes Z)$. A similar circuit provides the gradient.}
\label {fig_light_cone}
\end{center}
\end {figure}

\section{Model, algorithm and method}
\subsection{The VQCS algorithm}
We study the VQCS algorithm that was put forward in \cite{hoetal} as an adaptation of the Quantum Approximate Optimization Algorithm (QAOA) \cite{farhi, farhi2} to quantum simulators. Generally, the strategy is to split a target Hamiltonian acting on $L$ qubits as 
\begin{align}
H_T(h) &= H_1 + h H_2,
\end{align}
where the ground state of $H_1$, $\ket{\psi_1}$, should be easy to initialize on a quantum computer or simulator. The VQCS ansatz consists in choosing a circuit depth $2p$ as well as $2p$ variational parameters $(\gamma_1, \beta_1, \dots \gamma_p, \beta_p)$ and writing
\begin{align}
\label{eq_algo}
\ket{\psi(\pmb{\gamma},\pmb{\beta})} = e^{-i \beta_p H_1} e^{-i \gamma_p H_2} \dots e^{-i \beta_1 H_1}e^{-i \gamma_1 H_2} \ket{\psi_1}.
\end{align}
We take the cost function to be the energy density
\begin{align}
\label{eq_cost_function}
F_L(h; \pmb{\gamma},\pmb{\beta}) = \frac{1}{L} \braket{\psi(\pmb{\gamma},\pmb{\beta}) | H_T(h) |\psi(\pmb{\gamma},\pmb{\beta}) },
\end{align}
where we have made explicit $L$, the number of qubits that $H_T(h)$ acts on. The VQCS algorithm proceeds by measuring the cost function on the quantum computer or simulator and feeding the input to an optimization routine running on a classical computer. Once the optimal value
\begin{align}
F_L(h; p) &= \min_{\pmb{\gamma},\pmb{\beta}} F_L(h; \pmb{\gamma},\pmb{\beta}) 
\end{align}
is found, the corresponding quantum states can be prepared at will using the optimal parameters $(\pmb{\gamma}_\text{opt}, \pmb{\beta}_\text{opt})$ and observables of interest can be measured.

\subsection{The transverse field Ising model}
In this paper, we choose as the target Hamiltonian the TFIM
\begin{align}
\label{eq_model}
H_1 &= - \sum_{i=1}^{L} Z_i Z_{i+1} \\
H_2 &= - \sum_{i=1}^{L} X_i,\label{hh}
\end{align}
where $X = \begin{pmatrix}
0 & 1 \\
1 &  0
\end{pmatrix}$ and $Z = \begin{pmatrix}
1 & 0 \\
0 & -1  
\end{pmatrix}$ and where periodic boundary conditions are imposed. To be precise, $H_1$ has two ground states and we pick $\ket{\psi_1} = \ket{0 \dots 0}$ with $Z\ket{0} = +\ket{0}$ to initialise the algorithm. The fact that $H_1$ and $H_2$ consist of local commuting terms has two consequences: i) the evolution~(\ref{eq_algo}) can be implemented on a digital quantum computer without trotterization and ii) we can simulate the algorithm directly in the thermodynamic limit. More precisely,
\begin{align}\label{eq_tdlimit}
F_L(h; \pmb{\gamma},\pmb{\beta}) &= F_{L'}(h; \pmb{\gamma},\pmb{\beta}) 
\end{align}
for $L,L' \geq 4p$. This is due to the fact that there is an exact finite light cone for local operators, as illustrated in Fig. \ref{fig_light_cone}. In the regime $L\geq 4p$ the boundary conditions also become  irrelevant. Therefore, for a given $p$, a quantum computer with $L = 4p$ qubits captures the thermodynamic limit exactly and this is the setting we will later choose in our simulation \cite{footnote}.

Besides the energy, we are also interested in the scaling behaviour of the observables
\begin{align}
m_Z(h; p) = \braket{\psi(\pmb{\gamma}_\text{opt}, \pmb{\beta}_\text{opt}) | Z_j |\psi(\pmb{\gamma}_\text{opt}, \pmb{\beta}_\text{opt}) },\\
m_X(h; p) = \braket{\psi(\pmb{\gamma}_\text{opt}, \pmb{\beta}_\text{opt}) | X_j |\psi(\pmb{\gamma}_\text{opt}, \pmb{\beta}_\text{opt}) },\\
m_{XX}(\ell;h; p) =\qquad\qquad\qquad\qquad\qquad\qquad\qquad \\
\braket{\psi(\pmb{\gamma}_\text{opt}, \pmb{\beta}_\text{opt}) | X_jX_{j+\ell} |\psi(\pmb{\gamma}_\text{opt}, \pmb{\beta}_\text{opt}) }-m_X(h;p)^2,
\end{align}
where the position $j$ is arbitrary due to translation invariance of both the initial state and the circuit.

\subsection{The exact solution of the quantum circuit \label{exact}}
General techniques suffer from a cost per optimization step that is exponential in the system size, limiting the availability of data to $L \lessapprox 20$. To study the fate of the algorithm closer towards the scaling limit, we turn to techniques from integrability. As we show in Appendix \ref{solution}, the energy $F$, magnetization $m$ and the overlap $\phi$ of $\ket{\psi(\pmb{\gamma},\pmb{\beta})}$ with the exact ground state of the target Hamiltonian admit closed form expressions. 

In order to present these results, let us introduce the sequence of functions $f_j(k)$ for $j=0,...,2p$ that satisfy the following recurrence relation
\begin{equation}\label{recf}
\begin{aligned}
f_0(k)&=0\\
f_{2j+1}(k)&=e^{-4i\gamma_{j+1}}\frac{1-i\tan(k/2)f_{2j}(k)}{f_{2j}(k)-i\tan(k/2)}\\
f_{2j}(k)&=e^{-4i\beta_{j}}\frac{1+i\tan(k/2)f_{2j-1}(k)}{f_{2j-1}(k)+i\tan(k/2)}\,.
\end{aligned}
\end{equation}
We then define
\begin{equation}\label{eps}
\begin{aligned}
f_{\rm proj}(k)&=\frac{iK_{h0}(k)+f_{2p}(k)}{1+iK_{h0}(k)f_{2p}(k)}\\
&\text{with }K_{h0}(k)=\tan \frac{1}{2}\Big[\arctan \frac{h-\cos k}{\sin k}\\
&\qquad\qquad\qquad\qquad+\arctan \frac{1}{\tan k} \Big]\\
\varepsilon_h(k)&=\begin{cases}2\sqrt{1+h^2-2h\cos k}\,,\,\,\, k\neq 0\\ -2(1-h)\,,\,\,\, k= 0\end{cases}\,,
\end{aligned}
\end{equation}
as well as the sets
\begin{equation}\label{nsr}
\begin{aligned}
   {\rm NS}&=\left\{\frac{2\pi(n+1/2)}{L},\, n=-L/2,...,L/2-1\right\}\\
{\rm R}&=\left\{\frac{2\pi n}{L},\, n=-L/2,...,L/2-1\right\}\,,
\end{aligned}
\end{equation}
and ${\rm NS}_+,{\rm R}_+\subset {\rm NS,R}$ the subsets of strictly positive elements.  
Then the energy density $F_L(h;\pmb{\gamma},\pmb{\beta})$ of the state $|\psi(\pmb{\gamma},\pmb{\beta})\rangle$ is
\begin{equation}\label{energy}
\begin{aligned}
F_L(h;\pmb{\gamma},\pmb{\beta})=&\frac{1}{2L}\sum_{k\in{\rm NS}\cup {\rm R}}\varepsilon_h(k)\frac{|f_{\rm proj}(k)|^2}{1+|f_{\rm proj}(k)|^2}\\
&-\frac{1}{4L}\sum_{k\in{\rm NS}\cup {\rm R}}\varepsilon_h(k)\,.
\end{aligned}
\end{equation}
The overlap $\phi(h;\pmb{\gamma},\pmb{\beta})$ is
\begin{equation}\label{overlap}
\begin{aligned}
\phi(h;\pmb{\gamma},\pmb{\beta})=&e^{-iL(\sum_{j=1}^p\beta_j+\gamma_j)}\frac{\phi^{\rm NS}(h;\pmb{\gamma},\pmb{\beta})}{\sqrt{2}}\\
\phi^{\rm NS}(h;\pmb{\gamma},\pmb{\beta})=&\prod_{j=0}^{p-1} \prod_{k\in{\rm NS}_+}\left[\sin \frac{k}{2}+i\cos \frac{k}{2} f_{2j}(k) \right]\\
&\qquad\qquad\times\left[\sin \frac{k}{2}-i\cos \frac{k}{2} f_{2j+1}(k) \right]\\
&\times\prod_{k\in{\rm NS}_+} \frac{1+iK_{h0}(k)f_{2p}(k)}{\sqrt{1+K^2_{h0}(k)}}\,.
\end{aligned}
\end{equation}
The expressions for the energy \eqref{energy} and overlap \eqref{overlap} hold in arbitrary finite size $L$.

In the thermodynamic limit $L\to\infty$, the $X$-magnetization is given by
\begin{equation}\label{magX}
\begin{aligned}
&m_X(h; p)=1-\frac{1}{\pi}\int_{-\pi}^\pi\frac{|g_{2p}(k)|^2}{1+|g_{2p}(k)|^2}\D{k}\,,
\end{aligned}
\end{equation}
with
\begin{equation}
    g_{2p}(k)=\frac{1-i\tan(k/2)f_{2p}(k)}{f_{2p}(k)-i\tan(k/2)}\,.
\end{equation}
The connected $XX$-magnetization correlation is given by
\begin{equation}\label{xx}
\begin{aligned}
    m_{XX}&(\ell;h; p)=\\
    &\left|\frac{1}{\pi}\int_{-\pi}^\pi \frac{g_{2p}(k)e^{-ik\ell}}{1+|g_{2p}(k)|^2} \right|^2-\left|\frac{1}{\pi}\int_{-\pi}^\pi \frac{|g_{2p}(k)|^2e^{-ik\ell}}{1+|g_{2p}(k)|^2} \right|^2\,.
\end{aligned}
\end{equation}
The $Z$-magnetization is expressed as a Fredholm determinant
\begin{equation}\label{mag}
    m_Z(h; p)=\det({\rm Id}-J)\,.
\end{equation}
Here, $J(\lambda,\mu)$ is the function defined on $[0,\pi]\times [0,\pi]$
\begin{equation}
\begin{aligned}
J(\lambda,\mu)&=\frac{2}{\pi}\frac{\rho(\lambda)\sin \lambda }{f_{2p}(\lambda)}\frac{1}{\cos \lambda-\cos \mu}\\
&\times\left[\dashint_0^\pi\frac{f_{2p}(k)\sin k}{\cos \lambda-\cos k}\D{k}-\dashint_0^\pi\frac{f_{2p}(k)\sin k}{\cos \mu-\cos k}\D{k}\right]\,,
\end{aligned}
\end{equation}
where $\dashint$ denotes an integral in principal value, and with
\begin{equation}
\label{last_eq}
    \rho(k)=\frac{1}{2\pi}\frac{|f_{2p}(k)|^2}{1+|f_{2p}(k)|^2}\,.
\end{equation}
Finally, if the algorithm is initialized in $\ket{+ \dots +}$ with $X\ket{+}=\ket{+}$, then the same relations hold true with $f_0(k)=\frac{1}{i\tan(k/2)}$ in \eqref{recf}, and with \eqref{energy} multiplied by $2$ and summation only in ${\rm NS}$. Crucially, all of these quantities can be computed in time and memory that scale as polynomials in $p$. Equations~(\ref{recf})-(\ref{last_eq}) are the central result of this paper. 

A few comments are in order. First, as shown in Appendix \ref{gaussiangates}, one can generalize the exact solution to allow for rounds of $\exp(it\sum_{j=1}^L Y_j Y_{j+1})$-gates in the quantum circuit \eqref{eq_algo}. Second, in the case where the algorithm is initialized in $\ket{+\dots +}$, it was conjectured in \cite{hoetal} that the ground state in finite size $L$ can be exactly prepared with $p=L/2$ steps. In our formalism, this conjecture translates to the statement that, for $f_0(k)=\frac{1}{i\tan(k/2)}$, the system of $L/2$ equations $f_{\rm proj}(k)=0\,, k\in {\rm NS}_+$ for $L$ real unknowns $\gamma_1,\beta_1,...,\gamma_p,\beta_p$ admits at least one solution. Third, as explained in Appendix \ref{quench}, as a byproduct of our calculations we obtain the solution to the long-standing problem of the full time-evolution of the order parameter after a quantum quench in the Ising model \cite{SPS:04,RSMS09,CEF10,CEF1,foini,blass,SE12,IR11,RI11,EFreview,delfino,collura,GFE}. Fourth, while we focus on the ground state in the main text, it is possible to target excited states using a variant of~(\ref{overlap}). Finally, to find the optimal solution, we use a Broyden–Fletcher–Goldfarb–Shanno algorithm, supplying the gradient of~(\ref{energy}), for which an analytic formula is available.

\section{Results}
All of the following results are given for the effective thermodynamic limit $L = 4p$.
\subsection{Energy \label{sec_energy}}
The optimal parameters $\pmb{\gamma},\pmb{\beta}$ are fed into~(\ref{energy}) to obtain the energy of the lowest-energy state that can be prepared with the quantum computer after $p$ steps. The energy density is then compared to the known value in the thermodynamic limit,
\begin{equation}
   F_\infty(h)= -\frac{1}{2\pi}\int_0^\pi \varepsilon_h(k)\D{k}\,.
\end{equation}
The results are plotted in Fig.\ref{energyfig} for three different values of the magnetic field $h=1$, $h<1$ and $h>1$. We observe the leading behaviour
\begin{equation}\label{scaling}
    \begin{aligned}
    F_L(h;\pmb{\gamma},\pmb{\beta})-F_\infty(h)=\begin{cases}
    A e^{-\lambda p}\,\quad \text{if }h<1\\
    \frac{\pi}{12 p^2}\,\quad \text{if }h=1\\
    \frac{B}{p}\,\quad \text{if }h>1\,,
    \end{cases}
    \end{aligned}
\end{equation}
with $A,B,\lambda$ some $h$-dependent constants. The factor $\frac{\pi}{12}$ is observed with precision $\sim 10^{-5}$.

We recall that the energy density of the exact ground state in finite-size $L$ would converge to its value in the thermodynamic limit exponentially fast for $h\neq 1$, while at $h=1$ the energy difference would have the leading behaviour $-\frac{\pi v_F c}{6L^2}$ with $v_F=2$ the Fermi velocity and $c=\frac{1}{2}$ the central charge of the Conformal Field Theory describing the critical point \cite{cardy,affleck,saleur}. We remark that for $h\leq 1$ the behaviour \eqref{scaling} is qualitatively the same in terms of the circuit depth $p$. In particular, the behaviour at the critical point could suggest a general leading order correction of $\frac{\pi v_F c}{12p^2}$.

However, for $h>1$ we observe a slowdown of the algorithm that bears similarities with adiabatic slowdown near criticality \cite{vandam}: adiabatically evolving across a critical point typically requires time $\sim 1/L^2$ if $1/L$ is the smallest gap of the system. Since the quantum computer is initialized in the symmetry-broken phase of the Ising model and the circuit respects the symmetry, a slowdown when targeting $h>1$ is expected. Conversely, when we initialize the computation in the paramagnetic $\ket{+\dots+}$-state, the reverse scaling \eqref{scaling} is observed: the energy density converges exponentially fast for $h>1$, as $\frac{\pi}{12p^2}$ for $h=1$ and is proportional to $1/p$ for $h<1$.

\begin{figure}[t]
\begin{center}
 \includegraphics[scale=0.56]{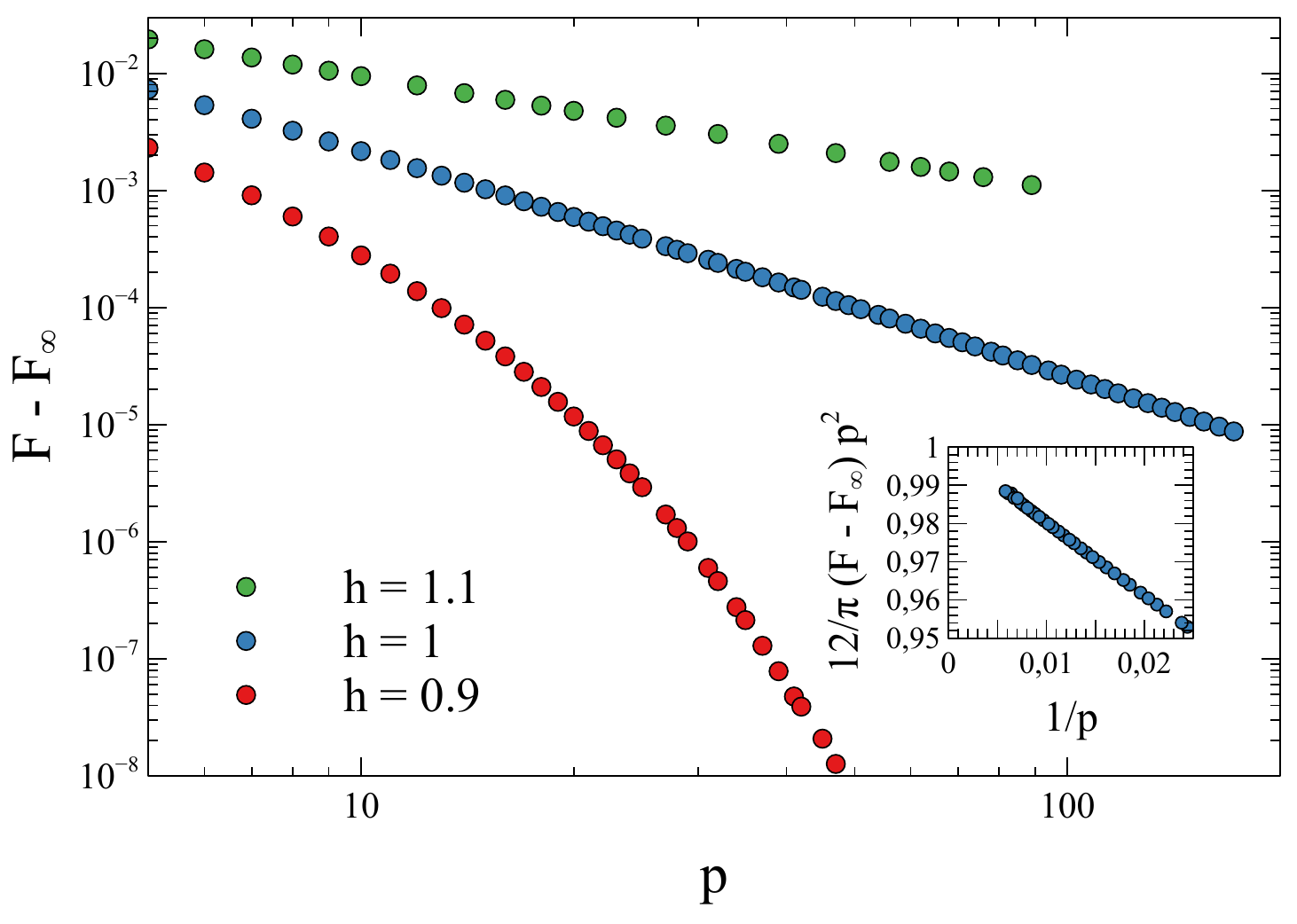}
\caption{Difference between the optimal energy density in the thermodynamic limit $F_L(h;\gamma,\beta)$ after $p$ steps and the exact ground state value, as a function of $p$, for $h=1$ (blue), $h=0.9$ (red) and $h=1.1$ (green). The inset shows the behaviour $\frac{\pi}{12p^2}$ for $h=1$.}
\label{energyfig}
\end{center}
\end{figure}

\subsection{$X$-magnetization}
We plot in Fig. \ref{fig_xmag} the expectation value $m_X(h;p)$ of the $X$-magnetization and its susceptibility $\chi_X(h;p)=\partial_h m_X(h;p)$ at optimal parameters $(\pmb{\gamma},\pmb{\beta})$ at circuit depth $p$, as a function of $h$. We observe a logarithmic divergence $\frac{2}{\pi} \log p$ of $\chi_X(h;p)$ at $h=1$, signaling a phase transition. The susceptibility is indeed known to diverge at the critical point with critical exponent $\alpha=0$ \cite{hamerbarber}.

We also show in Fig.~\ref{fig_xmag} the expectation value $m_{XX}(\ell;h;p)$ of the connected $XX$ magnetization correlation at the critical point $h=1$, as a function of the distance $\ell$. The exact known value is $\frac{4}{\pi^2} \frac{1}{4\ell^2-1}$ \cite{pfeuty}. With a circuit depth $p=188$, this exact value is well reproduced for $\ell\lessapprox 20$, and the critical exponent $2$ is observed up to $\ell \approx 30$. Our data suggests that the correlation length $\xi(p)$ behaves as $p/\pi^2$.

\begin{figure}[ht]
\begin{center}
 \includegraphics[scale=0.54]{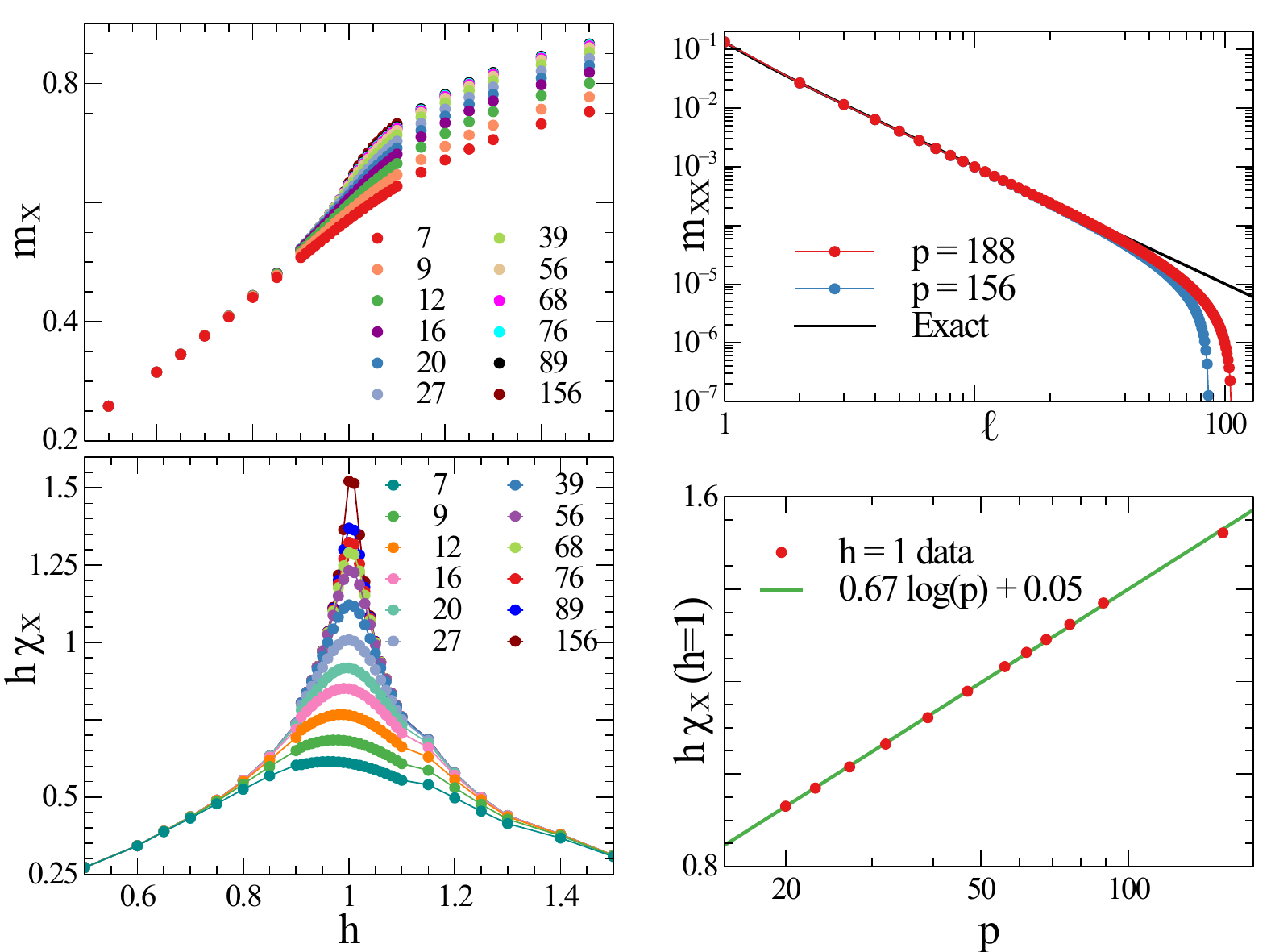}
\caption{Top left: magnetization in the $X$ direction at the optimal angles at circuit depth $p$ as a function of $h$. Top right: connected correlation of the magnetization in the $X$ as a function of $\ell$ for optimal angles at $h=1$. Bottom left panel: susceptibility in the $X$ direction at the optimal angles at circuit depth $p$ as a function of $h$. Bottom right: $\log$ divergence of the susceptibility at $h=1$. The lines guide the eye. }
\label{fig_xmag}
\end{center}
\end {figure}


\subsection{$Z$-magnetization and finite-circuit-depth scaling}

We now discuss the scaling of the order parameter, $m_Z(h;p)$. The exact value in the thermodynamic limit is known to be $(1-h^2)^{1/8}$ for $h\leq 1$ and $0$ for $h>1$ \cite{pfeuty}. The behavior of $m_Z(h;p)$ at optimal parameters $(\pmb{\gamma},\pmb{\beta})$, as shown in Fig. \ref{fig_zmag}, is qualitatively similar to the ground state magnetization in finite size $L$ due to  the finite light cone introducing a length scale $\leq 4p$. As we have observed in the preceding section, the correlation length $\xi(p)$ of the optimal states at circuit depth $p$ is significantly smaller. Nevertheless, since $\xi$ is linear in $p$, we expect the convergence in $p$ and $\xi$ to be characterised by the same exponents.

If the circuit reproduces the finite-size ground state sufficiently well, one may estimate the critical exponents by directly adapting finite-size arguments to obtain a scaling hypothesis for finite \emph{circuit depth} scaling: Denoting $\nu$ the critical exponent associated to the divergence of the correlation length $\xi \sim |h-h_c|^{-\nu}$ close to the critical point $h_c$, and $\beta$ the critical exponent of the $Z$ magnetization $m_Z(h;p)\sim |h-h_c|^\beta$, there exists a scaling function $\varphi$ such that for large $p$ and $h$ close to $h_c$
\begin{equation}
\label{eq_scaling_hypothesis}
    m_Z(h;p)p^{\beta/\nu}=\varphi((h-h_c)p^{1/\nu})\,.
\end{equation}

\begin{figure}[ht]
\begin{center}
 \includegraphics[scale=0.55]{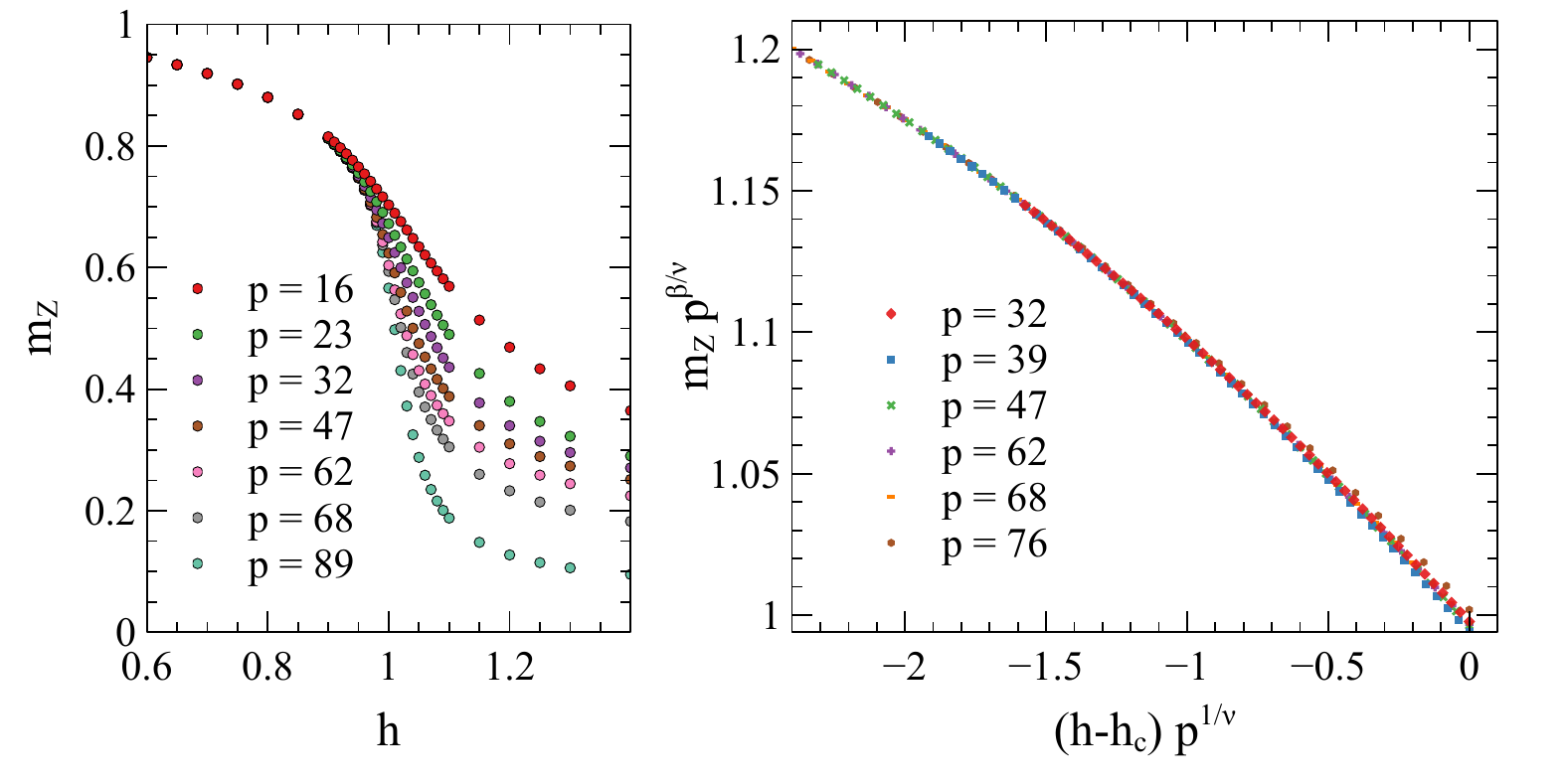}
\caption{Left: magnetization in the $Z$ direction at the optimal angles at circuit depth $p$ as a function of h. Right: one-sided collapse of the magnetization curves for 50 linearly spaced points between $h = 0.95$ and $h=1$, using the optimized exponents, see main text.}
\label{fig_zmag}
\end{center}
\end {figure}


There is however an important difference with finite-size scaling. We observed in Section \ref{sec_energy} that the convergence of the energy density as a function of $p$ is qualitatively similar to that in finite-size $L$ for $h<h_c$ and $h=h_c$ only. For $h>h_c$ the quantum circuit retains knowledge of the initial state, and the behaviour in terms of the circuit depth significantly differs. As a consequence one should generally expect finite-circuit-depth scaling to follow finite-size scaling only before the critical point. 
Having established $h_c=1$ from the $X$-susceptibility, optimizing the collapse yields $\beta = 0.122$ and $\nu=0.99$, in good agreement with the known values $\beta_\text{exact} = 0.125$, $\nu_\text{exact} = 1$, and such a one-sided collapse is indeed observed for those values see Fig. \ref{fig_zmag}.

\subsection{Structure of the optimal solutions}
\label{sec_structure_degeneracy}
\begin{figure}[t]
\begin{center}
 \includegraphics[scale=0.5]{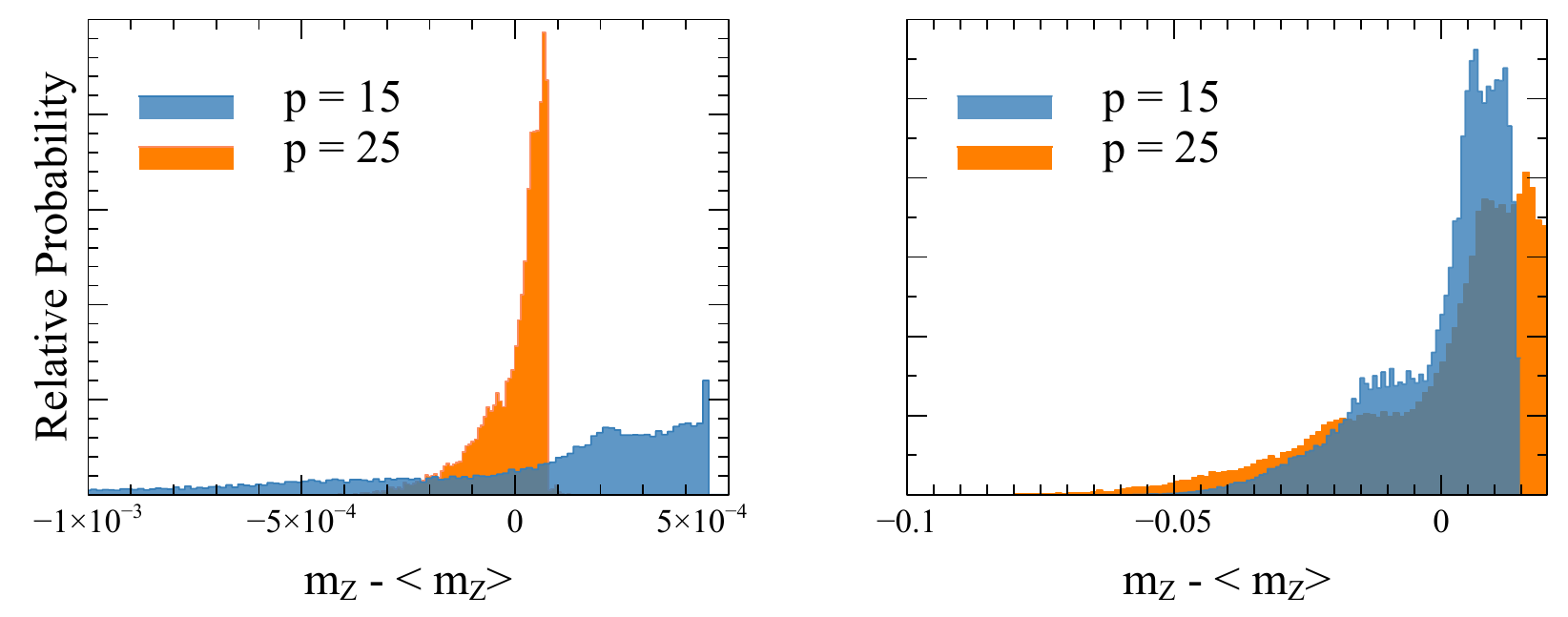}
\caption{Histogram of the $Z$-magnetization $m_Z$ of the optima obtained from 250'000 (80'000) random initializations of the optimization for $p=15$ ($p=25$) and targets $h=0.9$ (left) and $h=1.1$ (right). The angles are each initialised uniformly at random in $[0,\pi/2]$. The average over all samples is denoted by $\braket{m_Z}$. For $h=0.9$, $m_Z$ converges to the ground state magnetization, while, across the phase transition at $h=1.1$, such convergence is not observed for $p \leq 25$.}
\label{fig_histo}
\end{center}
\end{figure}

A key performance metric for variational quantum algorithms is the structure of the resulting energy landscape: an algorithm can only yield exponential speedup if the classical optimizer does not need to call the quantum subroutine too often and if local minima can be avoided. In this section we describe the structure of the solutions $(\pmb{\gamma},\pmb{\beta})$ that minimize the energy density in the thermodynamic limit at fixed $p$. We note that one can always add a multiple of $\pi/2$ to any of the angles without changing the state, so that one can choose they lie in $[0,\pi/2[$. We observe that for a circuit depth $p$, there are $2^{p}$ different sets of optimal angles $(\pmb{\gamma},\pmb{\beta})$ that yield exactly the same energy density in the thermodynamic limit $L\geq 4p$. This structure has been checked for $p\leq 5$, and we conjecture that it holds true for all $p$. One half of those minima are accounted for by the invariance of the energy under mapping $\gamma_j \rightarrow \pi/2 - \gamma_j$, $\beta_j \rightarrow \pi/2 - \beta_j$ for all $j$, as can be seen by direct inspection of~(\ref{energy}). Characterizing the remaining solutions is the goal of the remainder of this section.


Empirically, we find that the expectation values of the $X$-magnetization and the mode occupation numbers are identical in all of the ``branches". In particular the absolute value $|f_{\rm proj}(k)|$ is branch-independent. These observables share a common trait: after a Jordan-Wigner mapping from the spin to fermionic degrees of freedom (cf. Appendix~\ref{solution}), they correspond to quadratic fermion operators. Other observables such as the $Z$-magnetization or the $XX$ magnetization correlations vary among the different branches, cf. Fig. \ref{fig_histo}. In particular the phase of $f_{\rm proj}(k)$ depends on the branch. This branch dependence translates into the fact that at fixed finite $p$, the optimal energy for the TFIM Hamiltonian coupled to the observable of study would be discontinuous at vanishing coupling parameter.

For $h<1$, the distributions of expectation values become more and more peaked for $p \rightarrow \infty$ and converge to the same value among the different branches, reflecting the fact that all solutions converge to the same state. However, for $h>1$, our data suggests that this convergence could depend on the observable. While this convergence is observed to hold for the $XX$ magnetization correlation, our numerics suggest that, for $h>1$, the different branches could have different $Z$-magnetization even in the $p\to\infty$ limit. In fact, the behaviour of the determinant formula \eqref{mag} for the $Z$-magnetization (which is non-local in the fermions) in the limit $p\to\infty$ is rather subtle. For $h<1$ the function $f_{2p}(k)$ is optimized to approach a regular function of $k$, while for $h>1$ the target function is divergent at $k=0$, which makes \eqref{mag} approach a singular behaviour. 

We attribute this behaviour to the adiabatic slowdown discussed earlier. We could corroborate this hypothesis by initializing the algorithm in the $\ket{+\dots}$-state. However, the $Z$-magnetization of a state prepared from this initial state is exactly zero for symmetry reasons. Alternatively, one can initialize the algorithm in the state
\begin{equation}
    \frac{1}{\sqrt{2}} 
    \left(\ket{+ \dots +}+\frac{1}{\sqrt{L}}\sum_{j=1}^L \ket{+\dots+\underset{j}{-}+\dots+} \right) \,.
\end{equation}
Although this state may be difficult to prepare in practice, it allows for an exact treatment in our framework and yields a non-zero magnetization. In this case we observe indeed that the different branches for $h>1$ converge to the same magnetization in the limit $p\to\infty$.

To summarize this section, the energy landscape produced by the quantum circuit at depth $p$ is exponentially degenerate. The solutions corresponding to these minima have identical expectation values for observables that are quadratic in the fermions. Other observables differ between different optima. If the target Hamiltonian is in the same phase as the initial state, these expectation values converge to the ground state expectation value.

\subsection{Preparation time}
An attractive property of the VQCS algorithm is the fact that it is amenable to analog quantum simulation. In this case, the parameters $(\pmb{\gamma},\pmb{\beta})$ are interpreted as real times and the main criterion for the feasibility of the algorithm is the total time 
\begin{align}
\label{eq_Tofp}
T(p) = \sum_{j=1}^p \gamma_{\text{opt}}^j + \sum_{j=1}^p \beta_{\text{opt}}^j,
\end{align}
where all $\gamma_{\text{opt}}^j$ and $\beta_{\text{opt}}^j$ are assumed to be in $[0,\pi/2[$. If $T(p) \propto p$, then the same scaling relation~(\ref{eq_scaling_hypothesis}) hold, with $p$ replaced by $T$. Since the optimal solution is not unique (cf. Section \ref{sec_structure_degeneracy}), we take $T(p)$ to be the average total time over different ``branches". Numerically, we find that, on average, $\braket{\beta_{\text{opt}}^j}, \braket{\gamma_{\text{opt}}^j} \sim \pi/4$ such that $T(p) \sim \frac{\pi}{2} p$.
This is maximal in the sense that for every solution with $T(p) > \frac{\pi}{2} p$, there exists a solution with $T'(p) = \pi p - T(p) <  \frac{\pi}{2} p$, and so one can impose $T(p)\leq \frac{\pi}{2} p$.



\section{Discussion and Outlook}
We have shown that data obtained from small-scale quantum computers or simulators can be used in an informed way to accurately estimate the location and universality class of a quantum critical point. To this end, we have taken the perspective of quantum computation as a series of quenches and used techniques from integrability to obtain closed form expressions for the output of the VQCS algorithm. While these methods were necessary to argue about computations with more than $\sim 20$ qubits, they raise the question if similar scaling behaviour extends to non-integrable target Hamiltonians. We stress that the outcome of the quantum circuit is classically tractable independent of the target Hamiltonian, as long as the circuit is Gaussian. In this context, especially 2+1D quantum critical points would be interesting to investigate since a number of both digital and analog quantum devices naturally realize a two-dimensional architecture, while conformal data is hard to come by.

Regarding the slowdown of the algorithm across the phase transition, one may allow the circuit to break the symmetry of the target Hamiltonian, including gates like $\exp(i \delta Z)$. This may allow the computation to circumvent the critical point, potentially improving the scaling collapse~(\ref{eq_scaling_hypothesis}).

Naturally, circuits providing quantum advantage on NISQ devices will not be Gaussian and designing such circuits is a major challenge. The energy landscape may be sufficiently hostile that a random initialisation and optimization procedure is bound to fail; instead, usual contemporary approaches include only a small number of parameters relying on physical insight like the adiabatic theorem, coupled cluster wavefunctions, imaginary time evolution or dynamical mean-field theory \cite{vqe,vqaa, bauer_qalg, cerezo_variational, bharti_noisy, endo_hybrid, bassman_simulating,diniz_generalized,grimsley_adaptive, lee_generalized,mcardle_chemistry, colless_computation, higgott_excited, nakanishi_subspace, obrien_derivatives, mcardle_variational,motta_imaginary, manrique_periodic, rungger_meanfield, yoshioka_periodic, liu_periodic}. Adding a new, classically tractable ansatz to the toolkit may provide another useful starting point for quantum circuit design.

We also mention that, for large $p$, VQCS can mimic adiabatic evolution, yet the optimal evolution times we find at large $p$ do not generally bear the structure of trotterized adiabatic time evolution. Our techniques could be adapted to study the behaviour of quasi-adiabatic trajectories in the large-$p$ limit more generally.

Furthermore, it would be interesting to investigate the fate of this particular algorithm under noise. In the presence of an exponentially degenerate optimization landscape, small inaccuracies may even be beneficial by lifting degeneracies, while the linear depth of the algorithm circumvents an exponential slowdown due to noise-induced barren plateaus \cite{barren_plateaus_noise}. We leave these questions to future work.

\begin{acknowledgments}
We are grateful for discussions with F. H. L. Essler.
H.D. acknowledges support from the European Research
Council under the European Union Horizon 2020 Research and Innovation Programme via Grant Agreement
No. 804213-TMCS. E.G. acknowledges support from the EPSRC under grant EP/S020527/1.
\end{acknowledgments}
\begin{widetext}
\appendix
\section{The $X$ and $ZZ$ gates \label{solution}}
\subsection{The solution of the TFIM}
We first recall a number of results on the diagonalization of the TFIM, that can be found in Appendix A of \cite{CEF1} and that we briefly summarize here for self-completeness. The TFIM is defined as
\begin{equation}
H(h)=-\sum_{j=1}^L Z_jZ_{j+1}+hX_j\,,
\end{equation}
with periodic boundary conditions. It is diagonalized by performing a Jordan-Wigner transformation on the spin operators into fermions $\{c_j,c_l^\dagger\}=\delta_{jl}$
\begin{equation}\label{jw}
    X_j=1-2c_j^\dagger c_j\,,\qquad Z_j=(c_j+c_j^\dagger)\prod_{l=1}^{j-1}(1-2c_l^\dagger c_l)\,,
\end{equation}
followed by a Bogoliubov transformation of the Fourier modes
\begin{equation}\label{ck}
    c(k)=\frac{1}{\sqrt{L}}\sum_{j=1}^L e^{ijk}c_j\,,
\end{equation}
into fermionic operators $\{\alpha_{h;j},\alpha_{h;l}^\dagger\}=\delta_{jl}$
\begin{equation}\label{bog}
\begin{aligned}
c(k)&=\cos \frac{\theta^h_k}{2} \alpha_{h;k}+i \sin \frac{\theta^h_k}{2}\alpha^\dagger_{h;-k}\\
c^\dagger(k)&=-i\sin \frac{\theta^h_k}{2} \alpha_{h;-k}+\cos \frac{\theta^h_k}{2}\alpha^\dagger_{h;k}\,,
\end{aligned}
\end{equation}
where $\theta^h_k$ is defined by
\begin{equation}
e^{i\theta^h_k}=\frac{h-e^{ik}}{\sqrt{1+h^2-2h\cos k}}\,.
\end{equation}
The Hamiltonian conserves the parity of $\hat{N}=\sum_{j=1}^L c_j^\dagger c_j$, and so splits into two sectors, called Neveu-Schwartz (NS) and Ramond (R). If $\hat{N}$ is even the momentum $k$ in \eqref{ck} takes values in ${\rm NS}$ while if $\hat{N}$ is odd it takes values in ${\rm R}$, both sets being defined in \eqref{nsr}. Denoting $|0\rangle_h^{\rm R,NS}$ the vacuum state in these two sectors, the eigenstates of $H(h)$ then read, for $\pmb{k}\subset{\rm NS}$ or $\subset{\rm R}$
\begin{equation}
|\pmb{k}\rangle_h\equiv \alpha^\dagger_{h;k_{1}}...\alpha^\dagger_{h;k_{m}}|0\rangle^{\rm R,NS}_h\,,
\end{equation}
with $m$ even in the ${\rm NS}$ sector and odd in the ${\rm R}$ sector. The Hamiltonian $H(h)$ is then expressed as
\begin{equation}
    H(h)=\sum_{k\in{\rm NS}} \varepsilon_{h}(k) \alpha_{h;k}^\dagger \alpha_{h;k}+\sum_{k\in{\rm R}} \varepsilon_{h}(k) \alpha_{h;k}^\dagger \alpha_{h;k}-\frac{1}{2}\sum_{k\in {\rm NS,R}}\varepsilon_{h}(k)\,,
\end{equation}
with $\varepsilon_h(k)$ defined in \eqref{eps}, and where the last sum of this expression is performed on ${\rm NS}$ or ${\rm R}$ according to the sector of the states on which the Hamiltonian is applied. 

For $h>1$ the ground state is $|0\rangle^{\rm NS}_h$. For $h<1$ the two lowest-energy levels are $|0\rangle^{\rm NS}_h$ and $\alpha^\dagger_{h;0}|0\rangle^{\rm R}_h$ and their energy difference is exponentially small in the system size $L$. At $h=0$, the ground state is exactly twofold degenerate, and its corresponding eigenspace is generated by $|0\rangle ^{\otimes L}$ and $|1\rangle ^{\otimes L}$, with $|0\rangle$ and $|1\rangle$ the two eigenvectors of $Z$ with eigenvalues $1$ and $-1$. Their sum is in $\rm NS$ so is proportional to $|0\rangle^{\rm NS}_0$, and their difference is in $\rm R$  so is proportional to $\alpha^\dagger_{0;0}|0\rangle^{\rm R}_0$. Hence one can choose the arbitrary phase of the vacuum states so that the state $|\psi_1\rangle$ introduced below \eqref{hh} is
\begin{equation}
    |\psi_1\rangle=\frac{|0\rangle^{\rm NS}_0+\alpha^\dagger_{0;0}|0\rangle^{\rm R}_0}{\sqrt{2}}\,.
\end{equation}
As it will be useful, we will denote pair states in the NS sector by
\begin{equation}
|\pmb{\bar{k}}\rangle=|\pmb{k}\cup(-\pmb{k})\rangle\,,
\end{equation}
with $\pmb{k}\subset {\rm NS}_+$, and pair states in the R sector by
\begin{equation}
|\pmb{\bar{\bar{k}}}\rangle=|\pmb{k}\cup(-\pmb{k})\cup \{0\}\rangle\,,
\end{equation}
with $\pmb{k}\subset {\rm R}_+$.

\subsection{Coherent states}
Given a magnetic field $h$, a complex number $A$ and a function $f$ we define the so-called (fermionic) coherent states
\begin{equation}\label{coherent}
\begin{aligned}
    \Psi^{\rm NS}_h(A,f)&=A\sum_{\pmb{k}\subset {\rm NS}_+} \left[\prod_{k\in \pmb{k}}f(k)\right] |\pmb{\bar{k}}\rangle_h\\
    \Psi^{\rm R}_h(A,f)&=A\sum_{\pmb{k}\subset {\rm R}_+} \left[\prod_{k\in \pmb{k}}f(k)\right] |\pmb{\bar{\bar{k}}}\rangle_h\,.
    \end{aligned}
\end{equation}
The initial state $|\psi_1\rangle$ can be written as a sum of two coherent states in different sectors
\begin{equation}
    |\psi_1\rangle=\frac{\Psi^{\rm NS}_0(1,0)+\Psi^{\rm R}_0(1,0)}{\sqrt{2}}\,.
\end{equation}
The crucial observation of our paper that allows us to derive the exact formulas given in Section \ref{exact} is that a coherent state stays coherent if one changes the magnetic field $h$ to another magnetic field $\tilde{h}$. Namely we prove the following

\begin{property}\label{htildeh}
Let $h,\tilde{h}$ two arbitrary magnetic fields. We have\label{lem2}
\begin{equation}
\Psi^{\rm NS}_h(A,f)=\Psi^{\rm NS}_{\tilde{h}}(\tilde{A},\tilde{f})\,,
\end{equation}
with
\begin{equation}
\tilde{A}=A \prod_{k\in{\rm NS}_+}\frac{1+iK_{\tilde{h}h}(k)f(k)}{\sqrt{1+K^2_{\tilde{h}h}(k)}}
\end{equation}
and
\begin{equation}
\tilde{f}(k)=\frac{iK_{\tilde{h}h}(k)+f(k)}{1+iK_{\tilde{h}h}(k)f(k)}\,,
\end{equation}
where we defined
\begin{equation}
K_{\tilde{h}h}(k)=\tan \frac{\theta_k^{\tilde{h}}-\theta_k^{h}}{2}\,.
\end{equation}
We have an identical formula in the R sector.
\end{property}
\begin{proof}
We know the following relations between the vacuum states at $h$ and $\tilde{h}$, derived in \cite{CEF1}
\begin{equation}\label{vaccu}
|0\rangle_{h}^{\rm NS}=\prod_{k\in {\rm NS}_+}\left[\frac{1+iK_{\tilde{h}h}(k)\alpha_{\tilde{h};-k}^{\dagger}\alpha_{\tilde{h};k}^{\dagger}}{\sqrt{1+K^2_{\tilde{h}h}(k)}}\right]|0\rangle_{\tilde{h}}^{\rm NS}\,.
\end{equation}
as well as the relations between creation operators at different magnetic fields
\begin{equation}
\alpha_{h;k}= \cos \frac{\theta_k^{\tilde{h}}-\theta_k^{h}}{2}\alpha_{\tilde{h},k}+i\sin \frac{\theta_k^{\tilde{h}}-\theta_k^{h}}{2}\alpha^\dagger_{\tilde{h},-k}\,.
\end{equation}
This allows us to write for $\pmb{r}\subset {\rm NS}_+$
\begin{equation}
|\pmb{\bar{r}}\rangle_{h}=\frac{1}{\prod_{k\in {\rm NS}_+}\sqrt{1+K^2_{\tilde{h}h}(k)}}\prod_{r\in\pmb{r}}(iK_{\tilde{h}h}(r)+\alpha_{\tilde{h};-r}^{\dagger}\alpha_{\tilde{h};r}^{\dagger})\prod_{\substack{k\in{\rm NS}_+\\ k\notin\pmb{r}}}(1+iK_{\tilde{h}h}(k)\alpha_{\tilde{h};-k}^{\dagger}\alpha_{\tilde{h};k}^{\dagger})|0\rangle_{\tilde{h}}^{\rm NS}\,.
\end{equation}
We deduce from this the formula for the overlaps between two pair states at different magnetic fields
\begin{equation}
{}_{\tilde{h}}\langle \pmb{\bar{q}}|\pmb{\bar{r}}\rangle_{h}=\frac{\prod_{k\in \pmb{q}\perp \pmb{r}}iK_{\tilde{h}h}(k)}{\prod_{k\in {\rm NS}_+}\sqrt{1+K^2_{\tilde{h}h}(k)}}\,,
\end{equation}
where $ \pmb{q}\perp \pmb{r}=\pmb{q}\cup \pmb{r}-(\pmb{q}\cap \pmb{r})$. The overlap between $|\pmb{\bar{r}}\rangle_h$ and $|\pmb{q}\rangle_h$ is zero if $\pmb{q}$ is not a pair state.
We now obtain
\begin{equation}
\begin{aligned}
&\Psi^{\rm NS}_h(A,f)=\sum_{\pmb{q}\subset {\rm NS}}|\pmb{q}\rangle_{\tilde{h}}{}_{\tilde{h}}\langle \pmb{q}|\Psi^{\rm NS}_h(A,f)\\
&=A\sum_{\pmb{q}\subset {\rm NS}_+}\sum_{\pmb{r}\subset {\rm NS}_+}\left[\prod_{r\in \pmb{r}}f(r)\right]{}_{\tilde{h}}\langle \pmb{\bar{q}}|\pmb{\bar{r}}\rangle_h|\pmb{\bar{q}}\rangle_{\tilde{h}}\\
&=\frac{A}{\prod_{k\in {\rm NS}_+}\sqrt{1+K^2_{\tilde{h}h}(k)}}\sum_{\pmb{q}\subset {\rm NS}_+}\left[\prod_{q\in \pmb{q}}\left[iK_{\tilde{h}h}(q)\right]\sum_{\pmb{r}\subset {\rm NS}_+} \prod_{r\in \pmb{r}}\begin{cases}
\frac{f(r)}{iK_{\tilde{h}h}(r)} \qquad \text{if }r\in \pmb{q}\\
iK_{\tilde{h}h}(r)f(r) \qquad \text{if }r\notin \pmb{q}
\end{cases}\right] |\pmb{\bar{q}}\rangle_{\tilde{h}}\\
&=\frac{A}{\prod_{k\in {\rm NS}_+}\sqrt{1+K^2_{\tilde{h}h}(k)}}\sum_{\pmb{q}\subset {\rm NS}_+}\left[\prod_{q\in \pmb{q}}\left[iK_{\tilde{h}h}(q)\right] \prod_{q\in \pmb{q}}(1+\frac{f(q)}{iK_{\tilde{h}h}(q)}) \prod_{\substack{k\in {\rm NS}_+\\k\notin \pmb{q}}}(1+iK_{\tilde{h}h}(k)f(k))\right]|\pmb{\bar{q}}\rangle_{\tilde{h}}\\
&=\tilde{A}\sum_{\pmb{q}\subset{\rm NS}_+} \left[\prod_{q\in \pmb{q}}\tilde{f}(q)\right]|\pmb{\bar{q}}\rangle_{\tilde{h}}\,,
\end{aligned}
\end{equation}
with $\tilde{A},\tilde{f}$ defined in Lemma \ref{lem2}. 
\end{proof}

\subsection{The energy and the overlap}
We note that the time-evolution of a coherent state $\Psi^{\rm NS}_h(A,f)$ with the Hamiltonian $H(h)$ is directly given by
\begin{equation}
    e^{-itH(h)}\Psi^{\rm NS}_h(A,f)=\Psi^{\rm NS}_h(\tilde{A},\tilde{f})\,,
\end{equation}
with
\begin{equation}
    \tilde{f}(k)=e^{-2it\varepsilon_h(k)}f(k)\,,\qquad \tilde{A}=Ae^{-it \mathcal{E}^{\rm NS}}\,,
\end{equation}
and $\mathcal{E}^{\rm NS}$ the energy of the vacuum state $|0\rangle_h^{\rm NS}$. A similar relation holds in the R sector. Thus, applying $2p$ times Lemma \ref{htildeh} successively with $h=\infty$ and $h=0$, one finds
\begin{equation}
    |\psi(\gamma,\beta)\rangle=\frac{\Psi^{\rm NS}_0(A_{2p}^{\rm NS},f_{2p})+\Psi^{\rm R}_0(A_{2p}^{\rm R},f_{2p})}{\sqrt{2}}\,,
\end{equation}
with $f_{2p}$ given by recurrence \eqref{recf}, and $A_{2p}^{\rm NS,R}$ by
\begin{equation}
\begin{aligned}
A^{\rm NS,R}_{2p}=e^{-iL(\sum_{j=1}^p\beta_j+\gamma_j)}\prod_{j=0}^{2p-1}\prod_{k\in{\rm NS,R}_+}\left[\sin \frac{k}{2}+i(-1)^j \cos \frac{k}{2} f_j(k)\right]\,.
\end{aligned}
\end{equation}
In order to compute the expectation value of $H(h)$ of this state and the overlap with the ground state at magnetic field $h$, one performs another change of basis to $h$
\begin{equation}
     |\psi(\gamma,\beta)\rangle=\frac{\Psi^{\rm NS}_h(A_{\rm proj}^{\rm NS},f_{\rm proj})+\Psi^{\rm R}_h(A_{\rm proj}^{\rm R},f_{\rm proj})}{\sqrt{2}}\,,
\end{equation}
with $f_{\rm proj}(k)$ given in \eqref{eps}, and
\begin{equation}
    A_{\rm proj}^{\rm NS,R}=A_{2p}^{\rm NS,R}\prod_{k\in{\rm NS,R}_+}\frac{1+iK_{h0}(k)f_{2p}(k)}{\sqrt{1+K_{h0}(k)^2}}\,.
\end{equation}
Under this form, the formula for the projection onto the exact ground state given in \eqref{overlap} is readily deduced, as well as the energy \eqref{energy}. 

One notes that the overlap with the spontaneously symmetry broken ground state is given by
\begin{equation}
\begin{aligned}
\phi(h;\pmb{\gamma},\pmb{\beta})=&e^{-iL(\sum_{j=1}^p\beta_j+\gamma_j)}\frac{\phi^{\rm NS}(h;\pmb{\gamma},\pmb{\beta})+\phi^{\rm R}(h;\pmb{\gamma},\pmb{\beta})}{2}\,,
\end{aligned}
\end{equation}
with $\phi^{\rm R}$ defined similarly to $\phi^{\rm NS}$ with products in $\rm R_+$.

\subsection{$X$-magnetization}
The magnetization in the $X$ direction is the simplest to determine, since $X$ is local in terms of the fermions $c_j$. One first writes the state $|\psi(\gamma,\beta)\rangle$ as a sum of two coherent states at $h=\infty$. Using the Jordan-Wigner transformation \eqref{jw} and the Bogoliubov transformation \eqref{bog}, one has 
\begin{equation}
\begin{aligned}
\langle X_i\rangle=1-\frac{2}{L}\sum_{k\in{\rm NS}}\langle \alpha_{\infty;k}^\dagger \alpha_{\infty;k}\rangle\,,
\end{aligned}
\end{equation}
where $\langle .\rangle$ denotes an expectation value in the coherent state $\Psi_\infty^{\rm NS}(A,f)$. A similar formula holds in the $\rm R$ sector. Using then
\begin{equation}
    \langle \alpha_{0;k}^\dagger \alpha_{0;k}\rangle=\frac{|f(k)|^2}{1+|f(k)|^2}\,,
\end{equation}
one obtains formula \eqref{magX} in the thermodynamic limit. Formula \eqref{xx} for the connected correlation of the magnetization in the $X$ direction is obtained in a similar way.

\subsection{$Z$-magnetization}
The magnetization in the $Z$ direction is more delicate, since $Z$ is non-local in terms of the fermions. First, one notes that since the operator $Z$ changes the sector, one has
\begin{equation}\label{edtgy}
\begin{aligned}
    m_Z(h;p)&=\Re [\Psi^{\rm R}_0(A_{2p}^{\rm R},f_{2p})^\dagger Z_j\Psi^{\rm NS}_0(A_{2p}^{\rm NS},f_{2p})]\\
    &=\Re\left[ (A^{\rm R}_{2p})^*\sum_{\pmb{q}\subset {\rm R}_+}\prod_{q\in\pmb{q}}[f_{2p}(q)^*]F(\pmb{q})\right]\,,
\end{aligned}
\end{equation}
with
\begin{equation}
F(\pmb{q})={}^{\rm R}_0\langle \pmb{\bar{\bar{q}}}|Z_j \Psi^{{\rm NS}}_{0}(A_{2p}^{\rm NS},f_{2p})\,.
\end{equation}
Expressing $\Psi^{{\rm NS}}_{0}(A_{2p}^{\rm NS},f_{2p})$ in terms of energy eigenstates, one obtains a sum over the full Hilbert space of matrix elements of the $Z_j$ operator between two eigenstates, also called form factors. Their explicit expression is known and take a particularly simple form at $h=0$  \cite{Bugrij,BL03,Gehlen,iorgov11,CEF1}
\begin{equation}
{}_0^{\rm R}\langle \pmb{q}\cup\{0\}|Z_\ell|\pmb{k}\rangle_0^{\rm NS}=e^{i\ell(\sum_{q\in\pmb{q}}q-\sum_{k\in\pmb{k}}k)}\frac{(-i)^N}{L^{N}}\frac{\prod_{j<j'}\sin \frac{q_j-q_{j'}}{2}\prod_{j<j'}\sin \frac{k_j-k_{j'}}{2}}{\prod_{j,j'}\sin\frac{q_j-k_{j'}}{2}}\,,
\end{equation}
with $N$ the number of elements of $\pmb{q}$ and $\pmb{k}$. If this number differs in the two states then the form factor vanishes. In our case, the two states have to be pair states, and in this case one has the Cauchy determinant representation \cite{GFE}
\begin{equation}
\begin{aligned}
{}_{\rm R}\langle \pmb{\bar{\bar{q}}}|Z_\ell|\pmb{\bar{k}}\rangle_{\rm NS}=\frac{(-4)^N}{L^{2N}}(\det C)^2\prod_{j=1}^N \sin q_j\sin k_j\,,
\end{aligned}
\end{equation}
with $N$ the number of elements of $\pmb{q},\pmb{k}>0$, and with the matrix
\begin{equation}
C_{ij}=\frac{1}{\cos q_i-\cos k_j}\,.
\end{equation}
Let us fix a $\pmb{q}\subset {\rm R}_+$ with $N$ particles. One has then
\begin{equation}
F(\pmb{q})=\frac{A_{2p}^{\rm NS}}{N!}\frac{(-4)^N}{L^{2N}}\sum_{k_1,...,k_N\in {\rm NS}_+}(\det C)^2\prod_{j=1}^N [\sin q_j\sin k_jf_{2p}(k_j)] \,.
\end{equation}
We now prove the following Lemma \cite{GFE}, using techniques introduced in \cite{korepinslavnov}
\begin{property}\label{abab}
Given two functions $f(\lambda,\mu)$ and $g(\lambda,\mu)$, a set $K$ and two sets of numbers $\{\lambda_i\}_{i=1}^N,\{\mu_j\}_{j=1}^N$ we have the relation
\begin{equation}
\sum_{k_1,...,k_N\in K}\det_{i,j} \left[f(\lambda_i,k_j) \right]\det_{i,j} \left[g(k_i,\mu_j) \right]=N! \det_{i,j} \left[\sum_{k\in K} f(\lambda_i,k)g(k,\mu_j) \right]\,.
\end{equation}
\end{property}
\begin{proof}
We write
\begin{equation}
    \begin{aligned}
    \det_{i,j} \left[f(\lambda_i,k_j) \right]\det_{i,j} \left[g(k_i,\mu_j) \right]&=\det_{i,j} \left[\sum_{s=1}^N f(\lambda_i,k_s)g(k_s,\mu_j) \right]\\
    &=\sum_{\sigma\in\mathfrak{S}_N}(-1)^\sigma C'_{1\sigma(1)}...C'_{N\sigma(N)}\,,
    \end{aligned}
\end{equation}
with $C'_{ij}=\sum_{s=1}^N f(\lambda_i,k_s)g(k_s,\mu_j)$. Then
\begin{equation}
    \det_{i,j} \left[f(\lambda_i,k_j) \right]\det_{i,j} \left[g(k_i,\mu_j) \right]=\sum_{s_1,...,s_N\in\{1,...,N\}}\sum_{\sigma\in\mathfrak{S}_N}(-1)^\sigma f(\lambda_1,k_{s_1})g(k_{s_1},\mu_{\sigma(1)})...f(\lambda_N,k_{s_N})g(k_{s_N},\mu_{\sigma(N)})\,.
\end{equation}
Now, if $s_i=s_j$, changing $\sigma$ into $\sigma \cdot(ij)$ exactly multiplies the summand by $-1$, which makes this contribution vanish. Hence
\begin{equation}
    \det_{i,j} \left[f(\lambda_i,k_j) \right]\det_{i,j} \left[g(k_i,\mu_j) \right]=\sum_{\tau \in\mathfrak{S}_N}\sum_{\sigma\in\mathfrak{S}_N}(-1)^\sigma f(\lambda_1,k_{\tau(1)})g(k_{\tau(1)},\mu_{\sigma(1)})...f(\lambda_N,k_{\tau(N)})g(k_{\tau(N)},\mu_{\sigma(N)})\,.
\end{equation}
One can now perform the sum over $k_1,...,k_N$. The result is independent of $\tau$, which gives the $N!$ in the Lemma and the determinant formula.

\end{proof}

Using Lemma \ref{abab}, we obtain
\begin{equation}
F(\pmb{q})=A_{2p}^{\rm NS}\det B\,,
\end{equation}
with
\begin{equation}
\begin{aligned}
B_{ij}=\frac{4}{L^2}\sum_{k\in {\rm NS}_+}\frac{\sin q_i f_{2p}(k) \sin k }{(\cos q_i-\cos k)(\cos q_j-\cos k)}\,.
\end{aligned}
\end{equation}
In our case, $f_{2p}(k)$ obtained from \eqref{recf} is a regular function of $k$. Hence in the thermodynamic limit $L\to\infty$ we obtain, see \cite{GFE}
\begin{equation}
B_{ij}=f_{2p}(q_i)\delta_{ij}-\frac{2\sin q_i}{\pi L(\cos q_i-\cos q_j)}\left[\dashint_0^\pi\frac{f_{2p}(k)\sin k}{\cos q_i-\cos k}\D{k}-\dashint_0^\pi\frac{f_{2p}(k)\sin k}{\cos q_j-\cos k}\D{k}\right]+\mathcal{O}(L^{-2})\,.
\end{equation}
If $i=j$ the second term is understood as being the derivative obtained when $q_i\to q_j$. 

Let us denote $\rho(q)$ the so-called root density of $\pmb{q}$ when $L\to\infty$, namely the function such that $\rho(q)dq L$ is the number of elements in $\pmb{q}$ between $q$ and $q+dq$. It follows that one has the Fredholm determinant formula
\begin{equation}\label{interm}
F(\pmb{q})=A_{2p}^{\rm NS}\prod_{q\in\pmb{q}}[f_{2p}(q)]\det [{\rm Id}-J[\rho]](1+\mathcal{O}(L^{-1}))\,,
\end{equation}
with
\begin{equation}
J[\rho](\lambda,\mu)=\frac{2}{\pi}\frac{\rho(\lambda)\sin \lambda }{f_{2p}(\lambda)}\frac{1}{\cos \lambda-\cos \mu}\int_0^\pi\left[\frac{f_{2p}(k)\sin k}{\cos \lambda-\cos k}-\frac{f_{2p}(k)\sin k}{\cos \mu-\cos k}\right]\D{k}\,.
\end{equation}
We now go back to \eqref{edtgy}, that is a weighted sum of the root-density-dependent quantity $\det[{\rm Id}-J[\rho]]$ over all the eigenstates. We have the following Lemma, using ideas of \cite{EC}
\begin{property}\label{qa}
Let $F[\pmb{q}]$ be a function of $\pmb{q}$, and $f(k)$ a function. We define
\begin{equation}
     \langle F\rangle=\frac{1}{\prod_{k\in{\rm NS}_+}[1+|f(k)|^2]}\sum_{\pmb{q}\subset{\rm NS}_+}F[\pmb{q}]\prod_{q\in\pmb{q}}[|f(q)|^2]\,.
\end{equation}
If $F[\pmb{q}]=F[\rho]$ depends only on the root density $\rho$ of $\pmb{q}$ in the thermodynamic limit, then
\begin{equation}
    \langle F\rangle=F[\rho_s]+o(L^{0})\,,
\end{equation}
with
\begin{equation}
    \rho_s(k)=\frac{1}{2\pi}\frac{|f(k)|^2}{1+|f(k)|^2}\,.
\end{equation}
\end{property}
\begin{proof}
Let us first treat the particular case where in the thermodynamic limit $F$ depends only on $r$ the number of elements of $\pmb{q}$ divided by $L$. We introduce the generating function
\begin{equation}\label{alpha}
    \Gamma(\alpha)=\frac{1}{\prod_{k\in{\rm NS}_+}[1+|f(k)|^2]}\sum_{\pmb{q}\subset{\rm NS}_+}\prod_{q\in\pmb{q}}[(1+\tfrac{\alpha}{L})|f(q)|^2]\,.
\end{equation}
By differentiating with respect to $\alpha$, we see that
\begin{equation}
    \langle r^j\rangle=\Gamma^{(j)}(0)+\mathcal{O}(L^{-1})\,.
\end{equation}
Besides, performing the summation on $\pmb{q}$ we obtain
\begin{equation}
    \Gamma(\alpha)=\prod_{k\in{\rm NS}_+}\left[1+\frac{\alpha}{L} \frac{|f(k)|^2}{1+|f(k)|^2} \right]\,.
\end{equation}
From this we find for any $j$
\begin{equation}
  \langle r^j\rangle=\left(\int_0^\pi \rho_s(k)\D{k} \right)^j+\mathcal{O}(L^{-1})\,.
\end{equation}
As any regular function on $[0,\pi]$ can be approximated by a polynomial with arbitrary precision provided its degree is high enough, this establishes the result of the Lemma when $F$ is a function of $r$ only.

Let us now divide $[0,\pi]$ into $m$ windows $W_k=[\frac{\pi}{m}(k-1),\frac{\pi}{m}k]$ for $k=1,...,m$, and consider $F[r_1,...,r_m]$ a function of $\pmb{q}$ that  in the thermodynamic limit depends only on $r_k$'s, the number of elements of $\pmb{q}$ in $W_k$ divided by $L$. By introducing $\Gamma(\alpha_1,...,\alpha_m)$ as in \eqref{alpha} with $\alpha$ replaced by $\alpha_k$ where $k$ is such that $q\in W_k$, we get similarly
\begin{equation}
  \langle r_1^{j_1}...r_m^{j_m}\rangle=\left(\int_{ W_1} \rho_s \right)^{j_1}...\left(\int_{ W_m} \rho_s \right)^{j_m}+\mathcal{O}(L^{-1})\,.
\end{equation}
Hence the Lemma holds whenever $F$ is a function of $r_1,...,r_m$ only. Since any regular functional of $\rho$ can be approximated with arbitrary precision by such a function $F$ provided $m$ is large enough, the Lemma holds for general $F[\rho]$.
\end{proof}
Using \eqref{interm} in \eqref{edtgy} and Lemma \ref{qa}, one obtains the formula \eqref{mag} for the $Z$-magnetization in the thermodynamic limit.

\section{Coherent gates \label{gaussiangates}}
In this appendix we show that $\exp(it\sum_{j=1}^L Y_j Y_{j+1})$-gates can be incorporated in the quantum circuit considered in this paper, while still staying exactly solvable.
\subsection{Change of Pauli matrices}
Given a set of Pauli matrices $X,Y,Z$, the operators defined by the rotation
\begin{equation}
    X'=-X\,,\qquad Y'=Z\,,\qquad Z'=Y\,,
\end{equation}
give another set of Pauli matrices. The Ising Hamiltonian becomes
\begin{equation}
    H(h)=-\sum_{j=1}^L Y_j' Y_{j+1}'-hX'_j\,.
\end{equation}
Performing a Jordan-Wigner transformation as in \eqref{jw} with fermions $c'_j$, one finds the relation
\begin{equation}
    c'_j=i(-1)^j c^\dagger_j\,,
\end{equation}
which results in
\begin{equation}
    c'(k)=ic(\pi-k)^\dagger\,.
\end{equation}
This results in the following relation on the corresponding $\alpha'_{h;k}$ in \eqref{bog}
\begin{equation}
    \alpha'_{-\infty;-k}=-\sign(k)\alpha_{\infty;\pi-k}\,.
\end{equation}
We deduce that the coherent states $\Psi^{\rm NS,R \prime}_{h}(A,f)$ built from these new fermions satisfy
\begin{equation}
    \Psi^{\rm NS}_{\infty}(A,f)=\Psi^{\rm NS\prime}_{-\infty}(A,\tilde{f})\,,
\end{equation}
with
\begin{equation} \tilde{f}(k)=f(\pi-k)\,.
\end{equation}
An identical transformation holds in the R sector. By using Lemma \ref{htildeh}, one can thus transform a coherent state at magnetic field $h$ in the original set of Pauli matrices, into another coherent state at magnetic field $\tilde{h}$ in the new set of Pauli matrices. In particular at $h=0$, this allows us to apply a $\exp(it\sum_{j=1}^L Y_j Y_{j+1})$-gate to the state of the model.

\subsection{Set of coherent gates}
We obtain that the state of the quantum computer is exactly tractable whenever it is initialized in a superposition of coherent states \eqref{coherent} and time-evolved with the gates $\exp(it\sum_{j=1}^L X_j)$, $\exp(it\sum_{j=1}^L Y_j Y_{j+1})$ or $\exp(it\sum_{j=1}^L Z_j Z_{j+1})$ in any order. We call this set of gates "coherent gates" since they all map a coherent state onto a coherent state.

To present the transformation rules, without loss of generality one can assume by linearity that it is initialized in a unique coherent state, and since Lemma \ref{htildeh} shows that all magnetic fields $h$ are equivalent, one can assume that this coherent state is prepared with $h=\infty$. Hence we assume it is initialized in
\begin{equation}
    |\psi\rangle=\Psi^{\rm NS}_\infty(A,f)\,,
\end{equation}
with $A$ a complex number and $f$ a function on $[0,\pi]$. Under the application of any coherent gates the state stays coherent but its parameters are changed
\begin{equation}
    e^{it \sum_{j=1}^L \Gamma_j}|\psi\rangle=\Psi^{\rm NS}_\infty(\tilde{A},\tilde{f})\,,
\end{equation}
with
\begin{equation}
    \begin{aligned}
     \tilde{f}(k)&=e^{-4it}f(k)\,,\qquad   \text{if }\Gamma_j=X_j\\
     \tilde{f}(k)&=\frac{-i\tan(k/2)(1-e^{-4it})+(1+\tan^2(k/2)e^{-4it})f(k)}{\tan^2(k/2)+e^{-4it}+i\tan(k/2)(1-e^{-4it})f(k)}\,,\qquad   \text{if }\Gamma_j=Y_jY_{j+1}\\
     \tilde{f}(k)&=\frac{i\tan(k/2)(1-e^{-4it})+(1+\tan^2(k/2)e^{-4it})f(k)}{\tan^2(k/2)+e^{-4it}-i\tan(k/2)(1-e^{-4it})f(k)}\,,\qquad   \text{if }\Gamma_j=Z_jZ_{j+1}\,.
    \end{aligned}
\end{equation}

\section{$Z$-magnetization after a quantum quench \label{quench}}
In this appendix we show how the techniques developed in this paper can be applied to quantum quenches in the Ising model \cite{SPS:04,RSMS09,CEF10,CEF1,foini,blass,SE12,IR11,RI11,EFreview,delfino,collura,GFE}. This problem consists in initializing the state of the system $|\psi(t)\rangle$ in the ground state of $H(h_0)$ at magnetic field $h_0$, and to time-evolve it at $t>0$ with the Hamiltonian $H(h)$ at another magnetic field $h$.  

Using \eqref{vaccu} derived in \cite{CEF1}, one finds that the time-evolved state is
\begin{equation}
  |\psi(t)\rangle=\frac{\Psi^{\rm NS}_h(A^{\rm NS},f_t)+\Psi^{\rm R}_h(A^{\rm R},f_t)}{\sqrt{2}}\,,
\end{equation}
with
\begin{equation}
    A^{\rm NS,R}=\frac{1}{\prod_{k\in{\rm NS,R}_+}\sqrt{1+K^2_{hh_0}(k)}}\,,\qquad f_0(k)=iK_{hh_0}(k)e^{-2it\varepsilon_h(k)}\,.
\end{equation}
The difficulty encountered before is that the form factors of $Z_j$ between eigenstates of $H(h)$ are rather complicated \cite{Bugrij,BL03,Gehlen,iorgov11,CEF1} and do not allow for a resummation of the one-point function of $Z_j$ when expressed as a spectral sum. The idea is thus to perform a change of basis to $h=0$ , with Lemma \ref{htildeh}
\begin{equation}
  |\psi(t)\rangle=\frac{\Psi^{\rm NS}_0(\tilde{A}^{\rm NS},\tilde{f}_t)+\Psi^{\rm R}_0(\tilde{A}^{\rm R},\tilde{f}_t)}{\sqrt{2}}\,,
\end{equation}
with
\begin{equation}
 \tilde{A}^{\rm NS,R}_t=A^{\rm NS,R} \prod_{k\in{\rm NS,R}_+}\frac{1+iK_{0h}(k)f_t(k)}{\sqrt{1+K^2_{0h}(k)}}\,,\qquad \tilde{f}_t(k)=\frac{iK_{0h}(k)+f_t(k)}{1+iK_{0h}(k)f_t(k)}\,.
\end{equation}
Then one obtains exactly formula \eqref{mag} for the $Z$-magnetization after the quench, with $f_{2p}$ replaced by $\tilde{f}_t$.

\end{widetext}

\begin{thebibliography}{99}
\bibitem{footnote}
Empirically, we find that $L,L' \geq 2p + 2$ is sufficient for~(\ref{eq_tdlimit}) to hold.
\end{thebibliography}

\end{document}